\documentclass[letterpaper,10pt,conference]{ieeeconf}

\IEEEoverridecommandlockouts

\usepackage[utf8]{inputenc}                 % allow utf-8 input
\usepackage[T1]{fontenc}                    % use 8-bit T1 fonts
\usepackage[colorlinks=true]{hyperref}      % hyperlinks
\usepackage{url}                            % simple URL typesetting
\usepackage{booktabs}                       % professional-quality tables
\usepackage{amsfonts}                       % blackboard math symbols
\usepackage{nicefrac}                       % compact symbols for 1/2, etc.
\usepackage{microtype}                      % microtypography
\usepackage{xcolor}
\usepackage{amsfonts,amsmath,amssymb}
\usepackage{bbm}
\usepackage[ruled,vlined]{algorithm2e}
\usepackage{thmtools, thm-restate}
\usepackage{graphicx}
\usepackage{cite}
\usepackage{mathtools}

\usepackage[capitalise,nameinlink]{cleveref}    
\crefformat{equation}{#2(#1)#3}
\crefformat{figure}{#2Fig.~#1#3}
\crefformat{section}{#2\S#1#3}
\crefformat{subsection}{#2\S#1#3}
\crefformat{subsubsection}{#2\S#1#3}
\crefformat{assumption}{#2Assumption~#1#3}

\newcommand{\E}{\mathbb{E}}

\newcommand{\R}{\mathbb{R}}

\newcommand{\calX}{\mathcal{X}}
\newcommand{\calV}{\mathcal{V}}
\newcommand{\calF}{\mathcal{F}}

\newcommand{\calW}{\mathcal{W}}
\newcommand{\calU}{\mathcal{U}}
\newcommand{\calD}{\mathcal{D}}

\newcommand{\what}{\widehat{W}}
\newcommand{\calFhat}{\widehat{\mathcal{F}}}
\newcommand{\calDhat}{\widehat{\mathcal{D}}}

\newcommand{\iid}{\overset{\textup{iid}}{\sim}}

\newcommand{\wtilde}{\widetilde{W}}
\newtheorem{remark}{Remark}
\newtheorem{lemma}{Lemma}
\newtheorem{theorem}{Theorem}
\newtheorem{corollary}{Corollary}
\newtheorem{definition}{Definition}
\newtheorem{assumption}{Assumption}

\DeclareMathOperator*{\argmin}{arg\,min}

\DeclareMathOperator*{\minimize}{\mathrm{minimize}}
\DeclareMathOperator*{\subjectto}{\mathrm{subject~to}}

\title{\LARGE \bf%
Adaptive Robust Model Predictive Control\\with Matched and Unmatched Uncertainty
}
\author{Rohan Sinha, James Harrison, Spencer M. Richards, and Marco Pavone
\thanks{%
    The authors are with Stanford University, Stanford, CA. \texttt{\{rhnsinha, jharrison, spenrich, pavone\}@stanford.edu}
}%
}

\begin{document}

\maketitle
\thispagestyle{empty}
\pagestyle{empty}

%%%%%%%%%%%%%%%%%%%%%%%%%%%%%%%%%%%%%%%%%%%%%%%%%%%%%%%%%%%%%%%%%%%%%%%%%%%%%%%%
\begin{abstract}
    We propose a learning-based robust predictive control algorithm that compensates for significant uncertainty in the dynamics for a class of discrete-time systems that are nominally linear with an additive nonlinear component. 
    Such systems commonly model the nonlinear effects of an unknown environment on a nominal system. 
    We optimize over a class of nonlinear feedback policies inspired by certainty equivalent ``estimate-and-cancel'' control laws pioneered in classical adaptive control to achieve significant performance improvements in the presence of uncertainties of large magnitude, a setting in which existing learning-based predictive control algorithms often struggle to guarantee safety.
    In contrast to previous work in robust adaptive MPC, our approach allows us to take advantage of structure (i.e., the numerical predictions) in the a priori unknown dynamics learned online through function approximation. Our approach also extends typical nonlinear adaptive control methods to systems with state and input constraints even when we cannot directly cancel the additive uncertain function from the dynamics. Moreover, we apply contemporary statistical estimation techniques to certify the system's safety through persistent constraint satisfaction with high probability. Finally, we show in simulation that our method can accommodate more significant unknown dynamics terms than existing methods.
\end{abstract}

%*****************************************************************************
\section{Introduction}
Learning-based control offers promising methods to enable the deployment of autonomous systems in diverse, dynamic environments. Such methods learn from data to improve closed-loop performance over time. Upon deployment, these methods should provide safety guarantees and quickly adapt in the face of uncertainty; to this end, estimates of uncertainties in learned quantities must be maintained and updated as new data becomes available to ensure safety constraint satisfaction. However, many recently proposed learning-based control algorithms rely on uncertainty estimation methods that result in policies that are either too conservative (e.g., yielding limited performance to remain safe) or too fragile (e.g., infeasibility in the face of large uncertainties). In this work, we combine a simple nonlinear control law inspired by ``estimate-and-cancel'' methods in nonlinear adaptive control with robust model predictive control (MPC) techniques to control a system in an uncertain environment, represented using an unknown, nonlinear term in the dynamics. Our simulated examples show that this approach can both reduce the conservatism and fragility of existing methods. 

\textbf{Related Work.} 
We briefly review two significant paradigms for the control of uncertain systems, namely \emph{adaptive control} and \emph{robust control}. We then discuss recent works that combine ideas from both paradigms, oftentimes leveraging modern methods in machine learning.

\emph{Adaptive control} concerns the joint design of a parametric feedback controller and a parameter adaptation law to improve closed-loop performance over time when the dynamics are (partially) unknown \cite{slotine_applied_1991, ioannou_robust}. Design of these components for nonlinear systems commonly relies on expressing unknown dynamics terms as linear combinations of \emph{known basis functions}, i.e., \emph{features} \cite{slotine_applied_1991}. The adaptation law updates the feature weights online, and the controller applies part of the control signal towards cancelling the estimated term from the dynamics% as part of the full control input, applies an input to cancel the estimated unknown dynamics 
\cite{ioannou_adaptive_2006, lavretsky_adaptive, slotine_applied_1991, aastrom2013adaptive}. These simple methods can achieve tracking convergence up to an error threshold that depends on the representation capacity of the features relative to the true dynamics~\cite{ioannou_adaptive_2006, lavretsky_adaptive}. Recent works proposed combining high-capacity parametric and nonparametric models from machine learning with classical adaptive control designs. This includes deep neural networks via online back-propagation \cite{joshi_asynchronous_2020}, Gaussian processes \cite{chowdhary_bayesian_2015}, and Bayesian neural networks \cite{fan_bayesian_2020, harrison_meta-learning_2020} via online Bayesian updates and meta-learned features \cite{HarrisonSharmaEtAl2018, RichardsAzizanEtAl2021}. However, these approaches are fundamentally limited by common assumptions in classical adaptive control, namely that uncertain dynamics terms can be stably canceled by the control input in their entirety, i.e., that these terms are \emph{matched uncertainties} \cite{slotine_applied_1991, ioannou_adaptive_2006, ioannou_robust, lavretsky_adaptive}. Moreover, most of these works do not consider state and input constraints, essential to safe control in practice. We generalize these classical adaptive methods to incorporate safety constraints even if the uncertainty is not fully matched.

\begin{figure}
    \centering
    \includegraphics[width=.95\linewidth]{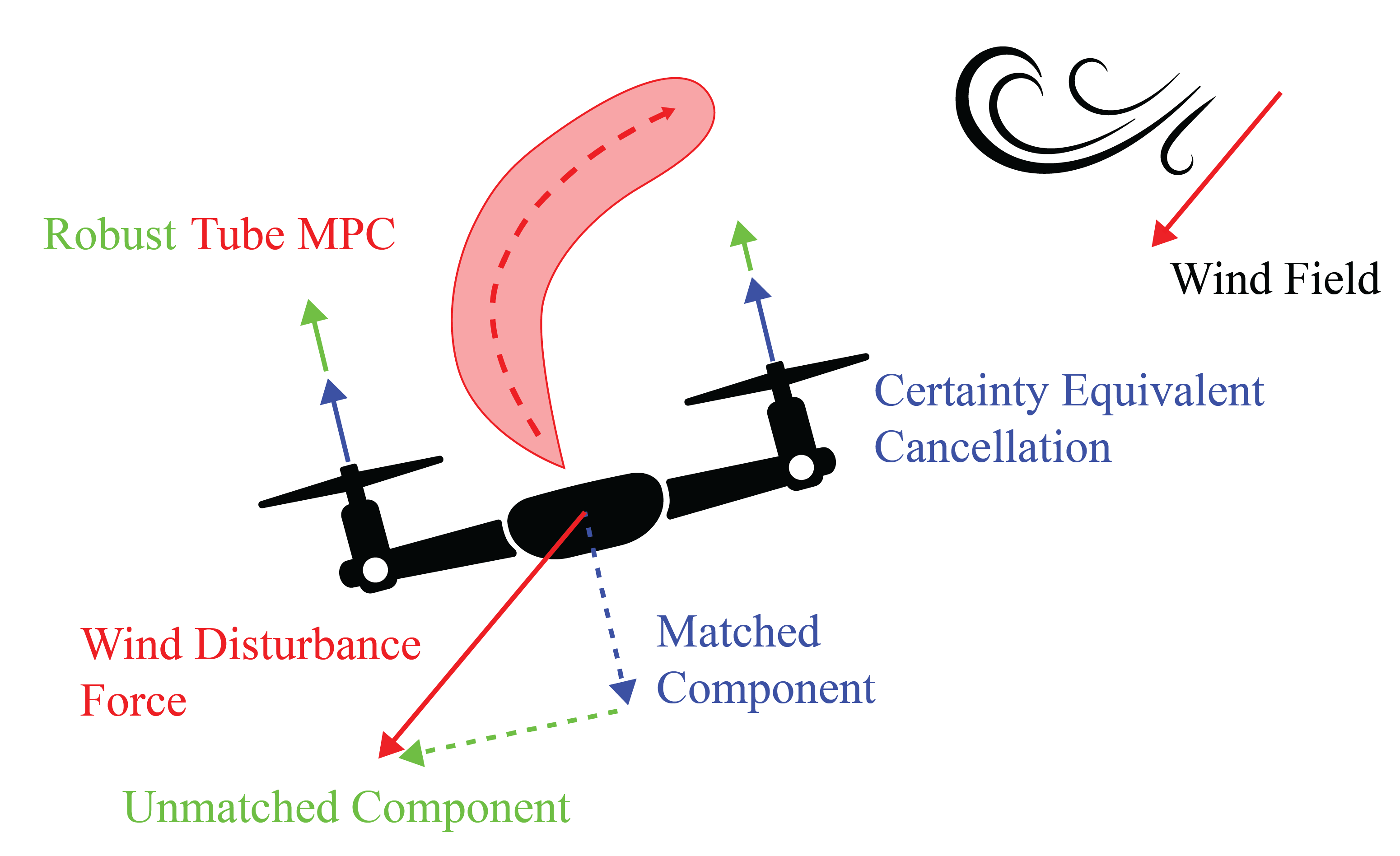}
    \caption{Illustration of our high-level approach: We learn the effect of environmental uncertainty online and curb its influence through certainty equivalent adaptive control. Then, we account for any remaining uncertainty and guarantee constraint satisfaction using tube MPC.}
    \label{fig:illustration}
\end{figure}

\emph{Robust control} seeks consistent performance despite uncertainty in the dynamics. In this work, we consider the robust control of constrained discrete-time systems using tools from predictive control. 
In particular, robust MPC algorithms for linear systems consider the control of a system subject to bounded noise or uncertain dynamics terms, i.e. \emph{disturbances}, as an optimization program with explicit state and input constraints. Some methods optimize the worst-case performance of the controller \cite{mayne:minmax}, while others tighten the constraints to accommodate the set of all possible trajectories induced by the disturbances and optimize the nominal predicted trajectory instead \cite{mayne_robust_2005, limon_input--state_2009}. To account for future information gain and reduce conservatism, these methods either fix a disturbance feedback policy \cite{mayne_robust_2005} or optimize over state feedback policies \cite{goulart_optimization_2006}. 

\emph{Adaptive robust MPC (ARMPC)}, often referred to as learning-based MPC, incorporates the online estimation (i.e., learning) of adaptive control methods into robust MPC to satisfy constraints in the presence of process noise and model uncertainty during learning.
Recent years have seen a flurry of work on nonlinear predictive control methods that apply contemporary machine learning techniques to learn uncertain dynamics online \cite{lew_safe_2020, HewingKabzanEtAl2017, koller_learning-based_2018, MishraGasparino2021}. These methods typically result in non-convex programs for trajectory optimization under the learned dynamics, while relying on conservative approximate methods for uncertainty propagation to guarantee constraint satisfaction. However, it is unclear how to construct the necessary components---the robust positive invariant and the terminal cost function---within predictive control to make claims of persistent constraint satisfaction (i.e., safety) or stability for arbitrary nonlinear systems. Some methods ignore these topics and do not make rigorous safety guarantees \cite{HewingKabzanEtAl2017}. Other work, such as \cite{koller_learning-based_2018, MishraGasparino2021}, assumes these ingredients already exist, or considers only trajectory optimization tasks where a goal region needs to be reached in a finite number of timesteps \cite{lew_safe_2020}. Moreover, iterative methods used to solve for local minima of non-convex programs can be prohibitively computationally expensive and often have limited associated performance guarantees. 

To make rigorous safety guarantees, we will focus on robust, adaptive methods for systems that are nominally linear as considered in \cite{bujarbaruah_adaptive_2018, BujarbaruahZhangEtAl2020, bujarbaruah_semi-definite_2020, soloperto_learning-based_2018, cairano_indirect_2016, KohlerAndinaEtAl2019, AswaniGonzalezEtAl2013}. A straightforward approach is to maintain an outer bound on an unknown, nonlinear term in the dynamics and use it as a disturbance bound in any chosen robust MPC scheme \cite{AswaniGonzalezEtAl2013, soloperto_learning-based_2018,  bujarbaruah_semi-definite_2020, bujarbaruah_adaptive_2018, BujarbaruahZhangEtAl2020}. These methods avoid some of the difficulties associated with trajectory optimization for nonlinear dynamics by ignoring the actual values of the nonlinear terms at any point in the state space. That is, these methods do not exploit the learned structure in the a priori unknown dynamics, often rendering them over-conservative or fragile.

\textbf{Contributions.}
We present an ARMPC method for systems with an additive unknown nonlinear dynamics term, subject to state and input constraints. Rather than construct an outer envelope for such terms, as is normative in ARMPC literature for linear systems, we develop theoretical guarantees for a broad class of \emph{function approximators}, including set membership and least-squares methods for certain noise models. Our key idea is to decompose uncertain dynamics terms into a \emph{matched} component that lies in a subspace that can be stably canceled by the control input, and an \emph{unmatched} component that lies in an orthogonal complement to this subspace. We apply certainty equivalent adaptive control techniques to stably cancel the \emph{matched} component from the dynamics and then apply robust MPC, considering the unmatched component as a bounded disturbance. Therefore, our method explicitly uses point estimates of the unknown term throughout the state space for control, i.e., it takes advantage of the learned structure in the dynamics.  

Our approach can be viewed through the lens of both adaptive control and robust adaptive MPC:
on one hand, we extend classical adaptive cancellation-based methods to a setting with uncertain, unmatched dynamics subject to state and input constraints.
On the other hand, we introduce a simple nonlinear feedback law to construct an adaptive, robust MPC that reduces the conservatism of existing approaches by taking advantage of the learned structure in a priori unknown dynamics. 
We prove our method is recursively feasible and input-to-state stable. Moreover, we demonstrate on various simulated systems that our method reduces the conservatism and increases the feasible domain of the resulting robust MPC problem compared to typical adaptive robust MPC methods. 

\section{Problem Formulation}\label{sec:prob}
We consider the robust control of nonlinear discrete time systems of the form
\begin{equation}\label{eq:problem-dynamics}
    x(t+1) = Ax(t) + Bu(t) + f(x(t)) + v(t),
\end{equation}
where~$x(t) \in \R^n$ is the system state, $u(t) \in \R^m$ is the control input, $A \in \R^{n \times n}$ and~$B \in \R^{n\times m}$ are known constant matrices, and $v(t) \in \calV$ is a disturbance in a known convex, compact set~$\mathcal{V}$ containing the origin. In addition, an \emph{unknown, nonlinear} dynamics term $f : \R^n \to \R^n$ acts on the system, representing the unmodelled influence of the environment on the nominally linear dynamics of system~\cref{eq:problem-dynamics}. For example, $f(x)$ can model the effect that wind conditions have on the linearized dynamics of a drone. We assume the disturbances are zero mean and independent and identically distributed (i.i.d.), i.e., $v(t) \iid p(v)$ and~$\E[v(t)] =0$ for all $t \geq 0$. Our goal is to regulate the system close to zero according to the robust optimal control problem
\begin{equation}\label{eq:problem}
\begin{aligned}
    \minimize_{x,u}\enspace&
        \E\Big[ \sum_{t=0}^\infty h(x(t), u(t)) \Big] 
    \\
    \subjectto\enspace
        &x(t+1) = Ax(t) + Bu(t) + f(x(t)) + v(t)
        \\&u(t) \in \calU,\, x(t) \in \calX,\, v(t) \in \mathcal{V}, \quad \forall t\in\mathbb{N}_{\geq 0}
        % \\, \quad\quad\quad\quad\quad\,\,\,\, \forall t \in \mathbb{N}_{\geq 0}
\end{aligned}
\end{equation}
where $\calX \subseteq \R^n$ and $\calU \subseteq \R^m$ are compact convex sets containing the origin,
and $h(x,u) =  x^\top Q x + u^\top R u$ is a quadratic stage cost parameterized by positive definite matrices $Q \in \mathbb{S}^n_{\succ 0}$ and $R \in \mathbb{S}^m_{\succ 0}$. The problem \cref{eq:problem} is computationally intractable to solve because the horizon is infinite and the nonlinear function $f$ makes the problem non-convex.
To approximately solve \cref{eq:problem}, we need additional assumptions on the unknown, nonlinear dynamics term $f$. In particular, to derive a controller that is robust to any possible value of $f(x)$, we need $f$ to be bounded on $\mathcal{X}$. 
Moreover, to construct guarantees on the online estimation of $f$ and establish properties of a controller using this estimate, we also need to assume some structure of $f$. For these reasons, we make the following assumption.

\begin{assumption}[structure]\label{as:bounded-feat}
    The nonlinear dynamics term $f : \R^n \to \R^n$ is linearly parameterizable, i.e.,
    \begin{equation}\label{eq:func-approx}
        f(x) = W\phi(x),\ \forall x \in \R^n,
    \end{equation}
    where $\phi : \R^n \to \R^d$ is a \emph{known} nonlinear feature map, and $W \in \R^{n \times d}$ is an \emph{unknown} weight matrix. Moreover, ${\|\phi(x)\| \leq 1}$ for any  $x \in \calX$, where $\|\cdot\|$ is the Euclidean~norm.
\end{assumption}

Representing a nonlinear function using a feature map is common both in adaptive control \cite{slotine_applied_1991, ioannou_robust} and contemporary machine learning \cite{harrison_meta-learning_2020, mania2020active}, as they can represent arbitrary functions if properly designed. For simplicity, we assume the features satisfy a unit norm bound without loss of generality; this still allows for function classes like neural networks with scaled sigmoid outputs.

\textbf{Matched and Unmatched Uncertainty:} While it is common in adaptive control to assume the uncertain function~$f$ in \cref{eq:problem-dynamics} can be stably cancelled in its entirety \cite{slotine_applied_1991, lavretsky_adaptive}, we will generalize this approach to a setting where perfect cancellation is not possible. To distinguish between the components of the uncertain dynamics~$f$ that can and cannot be cancelled, we classify the dynamic uncertainty as follows:
\begin{definition}[matched and unmatched uncertainty] 
    The uncertain function $f(x)$ in \cref{eq:problem-dynamics} is a \emph{matched uncertainty} if $f(x) \in \mathrm{Range}(B)$ for all $x \in \calX$. Conversely, if there exists an $x \in \calX$ such that $f(x) \notin \mathrm{Range}(B)$, then $f(x)$ is an \emph{unmatched uncertainty}.
\end{definition}

In this work, we assume that the $B$ matrix in \cref{eq:problem-dynamics} has full column rank, i.e., there are no redundant actuators, thereby guaranteeing the existence of the Moore-Penrose pseudoinverse $B^\dagger = (B^\top B)^{-1}B^\top$. Therefore, if $f(x)$ is a matched uncertainty, then the function $g(x) = B^\dagger f(x)$ satisfies $f(x) = Bg(x)$ for any $x \in \calX$. 

Controlling systems with matched uncertainty is a classical problem in the adaptive control literature, much of which relies on the observation that setting $u(t) = \bar{u}(t) - g(x(t))$ in~\cref{eq:problem-dynamics} would cancel the nonlinear term to yield linear dynamics with respect to the nominal input $\bar{u}(t)$. \emph{Certainty equivalent} controllers that approximately cancel $g(x)$ with an estimate $\hat{g}(x)$ result in simple nonlinear adaptive laws that can achieve asymptotic tracking performance for matched systems. Even though systems are often designed to be easy to control, unmatched uncertainty affects many practical systems of interest, such as the dynamics of any underactuated robotic system (e.g., quadrotors and cars). We propose to decompose the uncertain function $f$ into a matched and unmatched component, apply certainty equivalent cancellation to the matched component, and curb the impact of the unmatched component using robust MPC strategies. Applying part of the input to cancel matched uncertainty allows us to instantaneously prevent some components of $f$ from leaking into the dynamics, avoiding the need to react to large observed disturbances.

\textbf{ISS Stability.} Due to the disturbance $v(t)$, the system~\cref{eq:problem-dynamics} cannot be regulated exactly to the origin even asymptotically. Therefore, robust control algorithms are typically analyzed using Input-to-State Stability (ISS) theory to establish more appropriate stability properties \cite{limon_input--state_2009, goulart_optimization_2006, bujarbaruah_semi-definite_2020}. 
We briefly review relevant results in ISS theory for time-varying systems, since online adaptation of an estimate of the unknown function $f$ results in a time-varying closed loop system.

\begin{definition}[ISS stability \cite{li_input--state_2018}]
    The system $x(t+1) = q(t, x(t), v(t))$ with disturbance~$v(t)$ is globally \emph{Input-to-State Stable (ISS)} if there exist a class-$\mathcal{KL}$ function $\beta: \R_+ \times \R_+ \to \R_+$ and a class-$\mathcal{K}$ function $\gamma: \R_+ \to \R_+$ such that
    \begin{equation}
        \|x(t)\| \leq \beta(\|x(0)\|, t) + \gamma(\textstyle{\sup_{k \in \{0, \dots, t-1\}}} \|v(k)\|),
    \end{equation}
    for all $x(0) \in \R^n$ and $t \in \mathbb{N}_{\geq 0}$.
\end{definition}

In essence, ISS requires that the nominal system is asymptotically stable and the influence of the disturbances is bounded. This makes it a convenient framework to analyze the stability of systems subject to random disturbances. Similarly to regular nonlinear stability analysis, we can show a system is ISS if there exists an ISS-Lyapunov function. 

\begin{definition}[ISS-Lyapunov function]
    The function $V: \R \times\R^n \to \R$ is an \emph{ISS-Lyapunov function} for the system ${x(t+1) = q(t, x(t), v(t))}$ if it is continuous in $x$, continuous at the origin for all $t \in \mathbb{N}_{\geq 0}$, and there exist class-$\mathcal{K}_{\infty}$ functions $\alpha_1,\alpha_2,\alpha_3$ and a class-$\mathcal{K}$ function $\sigma$ such that
    \begin{equation}\begin{aligned}
        \alpha_1(\|x(t)\|) \leq V(t,x(t)) &\leq \alpha_2(\|x(t)\|) 
        \\
        V(t+1, x(t+1)) - V(t, x(t)) &\leq -\alpha_3(\|x(t)\|) + \sigma(\|v(t)\|)
    \end{aligned}
    \end{equation}
    for all $x(t) \in \R^n$.
\end{definition}

\begin{theorem}[\!\!\label{thm:iss}\cite{li_input--state_2018}]
    A time-varying system is globally ISS if it admits an ISS-Lyapunov function.
\end{theorem}

The above definitions naturally extend to local ISS stability; for a detailed discussion, we refer readers to \cite{li_input--state_2018,jiang_input--state_2001, limon_input--state_2009}.

%*****************************************************************************
\section{Adaptive Robust MPC}\label{sec:approach}
In this section we first describe assumptions on and necessary features of the learning procedure in a way that is agnostic to the choice of learning algorithm. We then introduce our adaptive robust MPC approach, and prove stability of the combined learning and control framework.

%------------------------------------------------
\textbf{Learning Desiderata.} Since the nonlinear dynamics term~$f$ is unknown, our method takes a \emph{certainty equivalent} approach by substituting an estimate~$\hat{f}$ that is refined online as more data becomes available. To guarantee the adaptive robust MPC framework satisfies state and input constraints for all time (i.e., safety), we make several assumptions on~$\hat{f}$.

We maintain the estimate
\begin{equation}
    \hat{f}(x,t) = \what(t)\phi(x)
\end{equation}
of $f(x)$, where $\what(t) \in \R^{n \times d}$ is our estimate of $W$ at time~$t$. To this end, we need bounds on our initial uncertainty, i.e., the difference between $f(x)$ and $\hat{f}(x,0)$ for all $x$. For a general statistical estimator, this entails specifying a risk tolerance $\delta \in (0,1)$ and computing confidence intervals on the estimate.

\begin{assumption}[prior knowledge]\label{as:prior-conf}
    Let $w_i$, $\hat{w}_i(t)$, and $\tilde{w}_i(t)$ be the $i$-th rows of $W$, $\what(t)$, and $\widetilde{W}(t) \coloneqq \what(t) - W$, respectively, for $i \in \{1,2,\dots,n\}$. At $t=0$, we know an initial estimate $\what(0) \in \R^{n \times d}$ and bounded sets $\{\mathcal{W}_i(0)\}_{i=1}^n$ with $\mathcal{W}_i(0) \subset \R^d$, such that $\widetilde{W}(0) \in \mathcal{W}(0)$ with probability at least $1 - \delta$, where
    \begin{equation}
        \mathcal{W}(t) \coloneqq \big\{
            \widetilde{W} \in \R^{n \times d} \mid \tilde{w}_i \in \mathcal{W}_i(t),\,\forall i \in \{1,\dots,n\}
        \big\},
    \end{equation}
    for all $t \in \mathbb{N}_{\geq 0}$.
\end{assumption}

\cref{as:prior-conf} provides only an initial bound on the error of the estimate, which we explicitly label as the estimate at $t=0$. Later, we will define $\mathcal{W}(t)$ for all $t \in \mathbb{N}_{\geq 0}$ when we adaptively update our estimate $\what(t)$ and the bounds $\{\mathcal{W}_i(t)\}_{i=1}^n$ online. Our approach leverages the certainty equivalent ``estimate and cancel'' control laws pioneered in classical unconstrained adaptive control \cite{slotine_applied_1991, lavretsky_adaptive}. As such, \cref{as:bounded-feat} and \cref{as:prior-conf} are necessary to bound our approximation error. However, since we specify the risk tolerance $\delta$, we cannot guarantee exact constraint satisfaction for all time. We instead relax our definition of safety to
\begin{equation}\label{eq:safety}
    \mathrm{Prob}\big(x(t) \in \calX,\, u(t) \in \calU,\, \forall t\geq 0\big) \geq 1 - \delta,
\end{equation}
which states that the probability of a constraint violation should be no more than $\delta$ over the entire realized trajectory.
Moreover, to guarantee closed loop safety, we assume we have an online adaptation strategy that ensures the quality of the estimate $\what(t)$ cannot get worse over time. 

\begin{assumption}[online learning]\label{as:shrinking}
    We have an \emph{online parameter estimator} that maps an initial estimate $\what(0)$, the associated $1 - \delta$ confidence interval $\calW(0)$, and the trajectory history $\{x(k), u(k)\}_{k=0}^t$ to an \emph{online estimate}~$\what(t)$ and confidence interval $\calW(t)$ at time $t$, such that $\hat{w}_i(t) - w_i \in \calW_i(t)$ for all time $t \in \mathbb{N}_{\geq 0}$ and $i \in \{1,2\dots,n\}$ with probability at least $1-\delta$. We assume the confidence intervals on $\what(t)$ are not growing with time, i.e.,%
    \begin{equation}
        \calW(t+1) \subseteq \calW(t),
    \end{equation}
    for all $t \in \mathbb{N}_{\geq 0}$.
\end{assumption}

Formulating separate confidence intervals for each row of $\what(t)$ is a natural approach, as fitting $\what(t)$ to historical data decomposes into $n$ separate least-squares problems (one for each row).
Crucially, \cref{as:shrinking} allows us to treat the confidence intervals $\calW(t)$ as exact bounds in the control design, since a controller that guarantees constraint satisfaction conditioned on the event that $\hat{w}_i(t) - w_i \in \calW_i(t)$ for all time then satisfies \cref{eq:safety}. Therefore, we treat the chance constraint \cref{eq:safety} as a proxy for robust constraint satisfaction and construct our approach for the remainder of \cref{sec:approach} conditioned on the event that $\hat{w}_i(t) - w_i \in \calW_i(t)$ for all time. This approach was also taken in \cite{lew_safe_2020, dean_regret_2018}. We discuss two commonplace estimators that satisfy our desiderata in \cref{sec:est}. 

\cref{as:shrinking} does not require the range of the estimated function to shrink over time. Therefore, our estimation procedure differs from methods such as \cite{bujarbaruah_adaptive_2018, bujarbaruah_semi-definite_2020, soloperto_learning-based_2018} that refine a shrinking bound exclusively on the range of the estimate without taking direct advantage of the structure in the nonlinear dynamics.

%------------------------------------------------
\subsection{Certainty Equivalent Cancellation}\label{sec:unmatched}
We propose optimizing over feedback policies that cancel as much of the nonlinear term $f(x)$ as possible.
\begin{definition}
The set of \emph{matching Certainty Equivalent (CE) policies} is the time-varying function class whose elements $\pi: \calX \times \mathbb{N}_{\geq 0} \to \calU$ are of the form
\begin{equation}\label{eq:unmatched-ce-law}
    \pi(x(t), t) = u^\star(x(t), t) - B^\dagger \hat{f}(x(t), t)
    \smallskip
\end{equation}
\end{definition}
The matching CE policies simply project $\hat{f}(x(t), t)$ onto $\operatorname{Range}(B)$, and cancel out as much of the disturbance as we can in the Euclidean norm sense, since
\begin{equation*}
    B^\dagger \hat f(x(t),t) = \argmin_z \|Bz - \hat{f}(x(t), t)\|.
\end{equation*}
The matching CE law \cref{eq:unmatched-ce-law} results in the closed-loop dynamics%
\begin{equation}\label{eq:unmatched-error}
\begin{aligned}
    x(t+1) 
    &= Ax(t) + B\pi(x(t),t) + f(x(t)) + v(t) \\
    &= Ax(t) + Bu^\star(x(t),t) + d(t)
\end{aligned}~,
\end{equation}
where we define the compound disturbance term as
\begin{equation}\begin{aligned}
    d(t) 
    &\coloneqq v(t) + f(x(t)) - BB^\dagger\hat{f}(x(t),t) \\
    &= v(t)\begin{aligned}[t]
        &+ BB^\dagger(f(x(t)) - \hat{f}(x(t),t)) \\
        &+ (I - BB^\dagger)f(x(t)).
    \end{aligned}
\end{aligned}
\end{equation}
We have written $d(t)$ above with three terms to highlight that it is driven by the process disturbance~$v(t)$, the estimation error~$f(x) - \hat{f}(x)$, and the imperfect matching using~$B^\dagger$.

We now introduce two simple polytopic approximations to bound the support of the cancellation term in \cref{eq:unmatched-ce-law} and the terms that make up the compound disturbance $d(t)$. This will allow us to apply the certainty equivalent cancellation in \cref{eq:unmatched-ce-law} and handle $d(t)$ by appropriately selecting $u^\star(\cdot)$ using robust MPC. As mentioned in our learning desiderata, we assume $\hat{w}_i(t) - w_i \in \calW_i(t)$ for the rest of this section.
\begin{lemma}\label{lem:unmatched-support}
Consider online approximation of $f(x)$ with features satisfying \cref{as:bounded-feat}, an estimator satisfying \cref{as:shrinking}, and define the estimated support set as
\begin{equation}
    \calF(t) \coloneqq \{z \in \R^n : |z_i| \leq \|\hat{w}_i(t)\| + 2 \max_{\tilde{w}_i \in \calW_i(t)}\|\tilde{w}_i\|\}.
\end{equation}
Then, for all $x \in \mathcal{X}$ and $t,k \in \mathbb{N}_{\geq 0}$ it holds that
\begin{equation}
    \hat{f}(x, t + k), f(x) \in \calF(t).
\end{equation}
\end{lemma}
\begin{proof}
We show $\calF(t)$ over-approximates the range of values $z =\hat{f}(x, t+k)$ can take for any $x \in \calX$ and $k \geq 0$. Let $z = \hat{f}(x, t+k) = \widehat{W}(t+k)\phi(x)$. Then,
\begin{equation}\begin{aligned}
    |z_i|
    &= |\hat{w}_i(t+k)^\top\phi(x)| \\
    &\leq \|\hat{w}_i(t+k)\| \\
    &\leq \|\hat{w}_i(t)\| + \|\hat{w}_i(t+k) - \hat{w}_i(t)\| \\
    &\leq \|\hat{w}_i(t)\| + \|\hat{w}_i(t+k) - w_i\| + \|\hat{w}_i(t) - w_i\|.
\end{aligned}
\end{equation}
The shrinking confidence interval property from \cref{as:shrinking} gives $ \calW_i(t+k)\subseteq \calW_i(t)$ for $k\geq 0$, so
\begin{equation}
    |z_i| \leq \|\hat{w}_i(t)\| + 2\max_{\tilde{w}_i \in \calW_i(t)}\|\tilde{w}_i\|,
\end{equation}
which proves that~$\hat{f}(x, t+k) \in \calF(t)$ for all~$k\geq 0$.
In addition, let~$y = f(x) = W\phi(x)$ for some $x\in\calX$. Then, 
\begin{equation}\begin{aligned}
    |y_i| 
    &= |w_i^\top \phi(x)| \\
    &\leq \|w_i\| \\
    &\leq \|\hat w_i(t)\| + \|\hat w_i(t) - w_i\| \\
    &\leq \|\hat w_i(t)\| + \max_{\tilde{w}_i \in \calW_i(t)}\|\tilde{w}_i\|.
\end{aligned}
\end{equation}
Hence,~$f(x) \in \calF(t)$. 
\end{proof}
% Proofs for all results are available in the extended version of this paper \cite{sinha2021adaptive}.

The set $\mathcal{F}(t)$ in \cref{lem:unmatched-support} contains all possible values that our online estimate can take for all future times. It is not straightforward to create a tighter approximation (i.e., eliminate the factor of 2) without additional assumptions. To see this, consider a constant unit norm ball confidence interval. In the worst case, the true parameter lies on the boundary of the ball around the current estimate. This means all future estimates may lie a Euclidean distance of 2 units away from the current estimate, yielding the bound in \cref{lem:unmatched-support}.

Moreover, \cref{lem:unmatched-support} does not require that $\calF(t+1) \subseteq \calF(t)$, so we provide the following corollary to help us create an approximation that is non-increasing in size.
\begin{corollary}\label{cor:unmatched-support}
    At time $t$, the sets $\{\calF(i)\}_{i=0}^t$ are known, so%\cref{lem:unmatched-support} implies 
    \begin{equation}
        \hat{f}(x, t + k), f(x) \in \bigcap_{i=0}^t \calF(i) \eqqcolon \calFhat(t)
    \end{equation}
    for all $x \in \calX$ and $k \in \mathbb{N}_{\geq 0}$, where we define $\calFhat(t)$ as the set
    \begin{equation}
        \{z : |z_i| \leq \min_{j \in \{0,\dots,t\}}[\|\hat{w}_i(j)\| 
            + 2\!\max_{\tilde{w}_i \in \calW_i(j)}\|\tilde{w}_i\|]\}.
    \end{equation}
    Note $\calFhat(t)$  can be computed recursively in time.
\end{corollary}
To construct a robust MPC problem to optimize the CE policy \cref{eq:unmatched-ce-law}, we need to account for the compound disturbance $d(t)$. We do this with the following lemma.
\begin{lemma}\label{lem:unmatched-disturbance}
Assume the online parameter estimator satisfies \cref{as:shrinking} with features that satisfy \cref{as:bounded-feat} and define the approximation error support $\mathcal{D}(t)$ as the set
\begin{align}
    \{z \in \R^n \mid |z_i| \leq \max_{\tilde{w}_i \in \calW_i(t)} \|\tilde{w}_i\|,  \forall i \in \{1, \dots,n\}\}.
\end{align}
If we control the system \cref{eq:problem-dynamics} using the certainty equivalent control law \cref{eq:unmatched-ce-law}, then at time $t$ for all $k\in \mathbb{N}_{\geq 0}$, the compound disturbance $d(t+k)$ in the dynamics~\cref{eq:unmatched-error} is contained in the set $\calDhat(t) \subseteq \calDhat(t-1)$, defined as
\begin{equation}\label{eq:unmatched-dist}
    \calDhat(t) \coloneqq (I - BB^\dagger)\calFhat(t) \oplus BB^\dagger \calD(t) \oplus \calV.
\end{equation}
Here $\oplus$ indicates the Minkowski sum and a matrix-set multiplication indicates a linear transformation of the set's elements.
\end{lemma}
\begin{proof}
At any state~$x \in \calX$ and time $t + k$, let $z = \hat{f}(x, t+k) - f(x) = (\what(t+k) - W)\phi(x) = \wtilde \phi(x)$ for some $\wtilde \in \calW(t+k)\subseteq \calW(t)$. Then
\begin{align}
    |z_i| &= |\tilde{w}_i^\top \phi(x_t)| \leq \|\tilde{w}_i\| \leq \max_{\tilde{w}_i \in \calW_i(t)} \|\tilde{w_i}\|. \nonumber
\end{align}
So $\hat{f}(x, t+k) - f(x) \in \mathcal{D}(t)$ for all $k \in \mathbb{N}_{\geq 0}$. Since $\calW_i(t) \subseteq \calW_i(t-1)$ by \cref{as:shrinking}, this implies $\mathcal{D}(t) \subseteq \mathcal{D}(t-1)$. Then, by \cref{cor:unmatched-support}, $f(x) \in \calFhat(t)$. Therefore,~$d(t)$ in the closed loop dynamics \cref{eq:unmatched-error} is contained in $\calDhat(t)$ since $\mathcal{D}(t)$ is symmetric. In addition, since both $\calFhat(t) \subseteq \calFhat(t-1)$ and~$\calD(t) \subseteq \calD(t-1)$, the support of $d(t)$ is nested over time, i.e., $\calDhat(t) \subseteq \calDhat(t-1)$.
% See the extended version of this paper \cite{sinha2021adaptive} for the proof.
\end{proof}
\begin{remark}\label{rem:proj}
We could consider multiple variations on the bounds in Lemmas \ref{lem:unmatched-support}, \ref{lem:unmatched-disturbance} that would yield equivalent properties of the closed loop system. 
% For example, we could maintain a separate bound on $f(\cdot)$ instead of the more conservative bound in \cref{lem:unmatched-support}. 
For example, if a bound on the true range of values of $f$ is known a priori, we may project $\hat{f}$ into a known box enclosing the support of $f$. 
\end{remark}

Rather than search over open-loop input sequences, which can incur issues with feasibility and stability under disturbances, we consider the standard practice of searching over closed-loop feedback policies \cite{goulart_optimization_2006, BujarbaruahZhangEtAl2020, mayne_robust_2005}. In particular, we follow \cite{goulart_optimization_2006} in optimizing over time-varying, causal, affine disturbance feedback policies of the form 
\begin{equation}\label{eq:unmatched-causal-pol}
    u_{t+k|t} = \bar{u}_{t+k|t} + {\textstyle\sum_{j=0}^{k-1}}K_{k,j|t}d_{t+j|t},
\end{equation}
to define the constraint-tightened robust MPC problem that we solve online as follows:
\begin{equation}\label{eq:unmatched-mpc}
\begin{aligned}
    \minimize_{\substack{
        \{K_{k,j|t}\}_{k=0,j=0}^{N-1,k-1},\\
        \{\bar{u}_{t+k|t}\}_{k=0}^{N-1}
    }}\enspace
    &V_N(\bar{x}_{t+N|t}) + \sum_{k=0}^{N-1} h(\bar{x}_{t+k|t}, \bar{u}_{t+k|t})
    \\
    \subjectto\enspace
    & \bar{x}_{t+k+1|t} = A\bar{x}_{t+k|t} + B\bar{u}_{t+k|t}
    \\& x_{t+k+1|t} = Ax_{t+k|t} + Bu_{t+k|t} + d_{t+k|t}
    \\& u_{t+k|t} = \bar{u}_{t+k|t}
        + {\textstyle\sum_{j=0}^{k-1}}K_{k,j|t}d_{t+j|t}
    \\& x_{k|t} \in \mathcal{X},\ u_{t+k|t} \in \mathcal{U} \ominus B^\dagger\calFhat(t)
    \\& \forall k \in \{0,1,\dots,N-1\}
    \\& \bar{x}_{t|t} = x(t),\ x_{t|t} = x(t),\ x_{t+N|t} \in \mathcal{O}(t)
    \\& \forall \{d_{t+k|t}\}_{k=0}^{N-1} \subset \widehat{\mathcal{D}}(t)
\end{aligned}~.
\end{equation}
The problem \cref{eq:unmatched-mpc} optimizes a time-varying feedback policy with a cost on the nominal trajectory $(\bar{x}, \bar{u})$ subject to state and input constraints on the realized trajectory. We use the subscript $t+k|t$ for quantities at the $k$-th step of the prediction horizon when \cref{eq:unmatched-mpc} is solved online at time~$t \in \mathbb{N}_{\geq 0}$. If the function $V_N : \mathcal{X} \to \R$, the terminal set $\mathcal{O}$, and the disturbance set $\calDhat$ are convex, then \cref{eq:unmatched-mpc} is a convex problem and we refer the reader to~\cite{goulart_optimization_2006} for implementation details. One might also consider a formulation of \cref{eq:unmatched-mpc} where the feedback gains are fixed, yielding a more basic tube MPC as considered in \cite{mayne_robust_2005}.

Let the optimal value of \cref{eq:unmatched-robust-pol} be $J^\star_N(t, x(t))$, with associated optimal policy sequence $[u^\star_{t|t}, \dots, u^\star_{t+N-1|t}]$. We then choose the robust control term of the \emph{certainty equivalent} control policy \cref{eq:unmatched-ce-law} as the receding horizon feedback law $u^\star : \mathcal{X} \times \mathbb{N}_{\geq 0} \to \mathcal{U}$ such that 
\begin{equation}\label{eq:unmatched-robust-pol}
    u^\star(x(t),t) = u_{t|t}^\star.
\end{equation}
 In \cref{eq:unmatched-mpc} we have tightened the input constraints to account for the matching term in the certainty equivalent policy \cref{eq:unmatched-ce-law} when compared to a more standard robust MPC problem (i.e. see \cite{goulart_optimization_2006}). As is standard in robust MPC, we assume we can both compute a robust control invariant set $\mathcal{O}(t)$ and have access to a convex terminal cost function $V_N$.

\begin{assumption}\label{as:unmatched-termcost}
    The terminal cost $V_N: \R^n \to \R_+$ is a continuous convex Lyapunov function for the nominal dynamics under a policy $u_N(x) = - Kx$. That is, there exists a class-$\mathcal{K}_\infty$ function $\alpha_N(\|x\|) \geq h(x, u_N(x))$ for which
    \begin{equation}
        V_N((A- BK)x) - V_N(x) \leq -\alpha_N(\|x\|).
    \end{equation}
\end{assumption}

\begin{assumption}\label{as:unmatched-termconst}
For the policy $u_N(x) = - Kx$ in \cref{as:unmatched-termcost}, the terminal set $\mathcal{O}(t) \subseteq \mathcal{X}$ is a maximal robust positive invariant set for the closed loop system $x(t+1) = (A - BK)x(t) + d(t)$ for $d(t) \in \calDhat(t)$ subject to $x(t) \in \calX$ and $u_N(x(t)) \in \calU \ominus B^\dagger \calFhat(t)$ for all $t\geq 0$. 
\end{assumption}

\cref{as:unmatched-termcost} and \cref{as:unmatched-termconst} are standard and easily satisfied by taking $u_N(x) = -Kx$ and $V_N(x) = x^\top P x$ as the solution to an LQR problem with $h$ as the stage cost. Then, $\mathcal{O}(t)$ can be computed efficiently using the standard algorithms in \cite{borrelli_predictive_2017}. 

\Cref{lem:unmatched-disturbance,lem:unmatched-support} imply that $\mathcal{O}(t-1) \subseteq \mathcal{O}(t)$ since $\calDhat(t) \subseteq \calDhat(t-1)$ and $\calFhat(t) \subseteq \calFhat(t -1)$. Therefore, the terminal constraint becomes less conservative over time. 

%------------------------------------------------
\subsection{Stability}\label{sec:properties}
We prove the stability of our algorithm through a recursive feasibility and input-to-state stability argument.
\begin{theorem}\label{thm:unmatched-recfeas}
Consider the system \cref{eq:problem-dynamics}, a parameter estimator that satisfies \cref{as:shrinking} with features that satisfy \cref{as:bounded-feat} in closed loop feedback with the \emph{matching certainty equivalent} control law \cref{eq:unmatched-ce-law},\cref{eq:unmatched-robust-pol}. If the tube MPC problem \cref{eq:unmatched-mpc} is feasible at~$t=0$, then for all~$t\geq 0$ we have that \cref{eq:unmatched-mpc} is feasible and the closed loop system \cref{eq:problem-dynamics},\cref{eq:unmatched-ce-law},\cref{eq:unmatched-robust-pol} satisfies~$x(t) \in \calX$, and $\pi(x(t),t) \in \calU$.
\end{theorem}
\begin{proof}
Suppose the optimal control problem \cref{eq:unmatched-mpc} is feasible at time $t$, with solution~$[u^\star_{t|t}(\cdot), \dots, u^\star_{t+N-1|t}(\cdot)]$. By \cref{lem:unmatched-support} and \cref{cor:unmatched-support} we have that $\hat{f}(x, t +k) \in \calFhat(t)$ for all $k \in \mathbb{N}_{\geq 0}$. Therefore, the time-varying CE control law
\begin{align}\label{eq:proof-pol}
    \pi_{t+k|t}(\cdot) = u^\star_{t+k|t}(\cdot) - B^\dagger \hat{f}(x(t), t + k)
\end{align}
satisfies the input constraints for $t \in [t, t+N-1]$, since \cref{eq:unmatched-mpc} then implies $u^\star_{t+k|t}(\cdot) \in \calU \ominus B^\dagger \calFhat(t)$. Moreover, by \cref{lem:unmatched-disturbance} the disturbance support shrinks in time, i.e., $\calDhat(t+1) \subseteq \calDhat(t) $. Therefore, we have that under policy \cref{eq:proof-pol} the closed-loop trajectory formed by \cref{eq:problem-dynamics},\cref{eq:proof-pol} satisfies $x(t+k) \in \calX$ for all $k \in [0,N]$ and that $x(t+N) \in \mathcal{O}(t)$. Hence, if we apply the CE policy \cref{eq:unmatched-ce-law},\cref{eq:unmatched-robust-pol} at time $t$, then $x(t+1) \in \calX$ and $\pi(x(t), t) \in \calU$.

By \cref{as:unmatched-termconst}, for any $x \in \mathcal{O}(t) \subseteq \mathcal{O}(t+1)$, applying the policy $u_N(x) \in \calU \ominus B^\dagger \calFhat(t)$ implies that $Ax + Bu_N(x) + d \in \mathcal{O}(t+1)$ for any $d \in \calDhat(t)$. Therefore, \cref{cor:unmatched-support} and \cref{lem:unmatched-disturbance} imply the policy sequence $[u^\star_{t+1|t}(\cdot), \dots u^\star_{t+N-1|t}(\cdot), u_N(\cdot)]$ is feasible for the tube MPC problem \cref{eq:unmatched-mpc} at time $t+1$. Therefore, if the MPC program \cref{eq:unmatched-mpc} is feasible at time $t=0$, it is also feasible for all $t\geq0$ and the closed-loop system formed by the matching CE law \cref{eq:problem-dynamics}, \cref{eq:unmatched-ce-law}, \cref{eq:unmatched-robust-pol} must robustly satisfy state and input constraints by induction.
\end{proof}
% As a reminder, all proofs are available in the extended version of this paper \cite{sinha2021adaptive}.

\begin{remark}
\cref{thm:unmatched-recfeas} and \cref{as:unmatched-termconst} imply that the maximal RPI set $\mathcal{O}(t)$ or disturbance sets $\calDhat(t), \calFhat(t)$ need not be updated at every timestep to guarantee recursive feasibility (nor stability), so we can apply this algorithm in an iterative setting and update the constraints only between episodes.
\end{remark}

\begin{theorem}\label{thm:unmatched-iss}
Consider a system of the form in \cref{eq:problem-dynamics}, 
%with \emph{matched uncertainty}, 
a parameter estimator that satisfies \cref{as:shrinking} with features that satisfy \cref{as:bounded-feat} in closed-loop feedback with the \emph{certainty equivalent} control law \cref{eq:unmatched-ce-law},\cref{eq:unmatched-robust-pol}. Let $\calX_N \subseteq \calX$ denote the set of states for which the tube MPC problem \cref{eq:unmatched-mpc} is feasible. Then the closed loop system is locally input-to-state stable with region of attraction $\calX_N$. 
\end{theorem}
\begin{proof}
Our proof closely resembles \cite[Thm.~2]{BujarbaruahZhangEtAl2020}. 
We argue that the nominal system is stable by a standard MPC argument, and that the closed-loop system is ISS since the disturbances are bounded.
Since we assume the stage cost is quadratic, there exist two class-$\mathcal{K}_\infty$ functions $\alpha_1, \alpha_2$ such that for all $t\geq0$, $\alpha_1(\|x\|) \leq J^\star_N(t,x) \leq  \alpha_2(\|x\|)$, a class-$\mathcal{K}_\infty$ function $\alpha_3$ such that $h(x,u) \geq \alpha_3(\|x\|)$, and  $J_N^\star(t,0)=0$ (see \cite[Prop.~1]{limon_input--state_2009}, \cite[Thm.~2]{BujarbaruahZhangEtAl2020}). Let $J^\star_N(t, x(t))$ be the solution of \cref{eq:unmatched-mpc} associated with the nominal prediction $[\bar{x}_{t|t}^\star, \dots, \bar{x}_{t+N|t}^\star]$ and feedback policies $[u^\star_{t|t}(\cdot), \dots, u^\star_{t+N-1|t}(\cdot)]$. As in the proof of \cref{thm:unmatched-recfeas}, we have that if we apply the CE control law \cref{eq:unmatched-ce-law},\cref{eq:unmatched-robust-pol} at time $t$, then the policies $[u^\star_{t+1|t}, \dots, u^\star_{t+N-1|t}, u_N]$ are a feasible solution for \cref{eq:unmatched-mpc} at time $t+1$. Let $\bar{J}(t, x)$ be the cost associated with forward simulating the nominal system using the policies $[u^\star_{t+1|t}, \dots, u^\star_{t+N-1|t}, u_N]$ with $x$ as initial condition.  i.e., set $u^\star_{t+N|t} = u_N(\cdot)$ and let $\bar{x}_{t+1|t+1} = x$, $\bar{x}_{k+1|t+1} = A\bar{x}_{k|t+1} + Bu^\star_{k|t}(\bar{x}_{k|t+1})$ for $k \in \{t+1, \dots, t+N\}$ so that
\begin{equation}
    \bar{J}(t,x) = \sum_{k=t+1}^{t + N}h(\bar{x}_{k|t+1}, u^\star_{k|t}(\bar x_{k|t+1})) + V_N(\bar{x}_{t+N+1|t+1}). \nonumber
\end{equation}
This gives that $J^\star_N(t+1, x(t+1)) \leq \bar{J}(t, x(t+1))$. Moreover, since the stage cost is quadratic and by \cref{as:unmatched-termcost}, $\bar{J}(t,x)$ is uniformly continuous in $x$ for all $t\geq 0$ on the state space since the inputs are constrained in a compact set. It follows that for $x_1, x_2 \in \calX$, there exists a $\mathcal{K}_\infty$ function $\alpha_J$ such that for all $t\geq0$, $|\bar{J}(t,x_1) - \bar{J}(t,x_2)| \leq \alpha_J(\|x_1 - x_2\|)$ (see \cite[Lem.~1]{limon_input--state_2009}). Therefore, 
\begin{equation}\begin{aligned}
        &J^\star_N(t + 1, x(t+1)) - J^\star_N(t, x(t)) 
        \\
        &\leq \bar{J}(t, x(t+1)) - J^\star_N(t, x(t)) 
        \\
        &= \bar{J}(t, x(t+1)) - \bar{J}(t, \bar{x}^\star_{t+1|t}) + \bar{J}(t, \bar{x}^\star_{t+1|t}) - J^\star_N(t, x(t)) 
        \\
        &\leq |\bar{J}(t, x(t+1)) - \bar{J}(t, \bar{x}^\star_{t+1|t})|
            - h(x_t, \bar{u}^\star_{t|t}(x(t))) 
        \\
        &\leq \alpha_J(\|d(t)\|) - \alpha_3(\|x(t)\|).
\end{aligned} \nonumber
\end{equation}
% Finally, since $d(t) \in \calDhat(t) \subseteq \calDhat(t-1)$ is contained in a closed and bounded set, there must exist a class-$\mathcal{K}$ function $\sigma$ such that $\alpha_J(\|d(t)\|) \leq \sigma(\|d(t)\|)$. Hence
% \begin{align*}
%     J^\star_N(t+1, x_{t+1}) - J^\star_N(t, x_t) \leq -\alpha_3(\|x_t\|) + \sigma(\|d_t\|)
% \end{align*}
So, the system is ISS by \cref{thm:iss}.
% See the extended version of this paper \cite{sinha2021adaptive} for the proof.
\end{proof}
\begin{remark}
The ISS result in \cref{thm:unmatched-iss} does not explicitly show that improvements in the confidence of the model lead to better performance of the controller, since we only assume the model confidence is non-decreasing in \cref{as:shrinking}. 
\end{remark}
In adaptive control, stronger guarantees of performance improvement are typically made under \emph{persistence of excitation} assumptions \cite{slotine_applied_1991}.

\section{Adaptation Laws \& Learning Algorithms}\label{sec:est}
In this section we highlight two common and perhaps complimentary online function approximation schemes, one is statistical and one is not, that satisfy the decaying confidence interval of \cref{as:shrinking} that we used to construct our ARMPC algorithm.

%------------------------------------------------
\textbf{Set Membership Estimation.} %\label{sec:setmembership}
A common approach in the adaptive MPC literature is to estimate constant, or slowly changing, disturbances through \emph{set membership} estimation \cite{bujarbaruah_adaptive_2018, KohlerAndinaEtAl2019}. These estimators maintain a feasible parameter set that is refined as more data becomes available. The feasible parameter set contains all credible model parameters that explain previous observations, which means that the feasible parameter sets are nested over time. We consider learning the parameters of a nonlinear uncertainty model of the form in \cref{eq:func-approx} directly using set-membership estimation. Under the prior knowledge \cref{as:prior-conf}, the initial feasible parameter set is given as $\Theta(0) = \{\what(0)\} \oplus \calW(0)$ and the feasible parameter set at time $t$ is obtained as
\begin{equation}\label{eq:est-setmember-fps}
\begin{aligned}
    \Theta(t) = &\big\{W \in \Theta :\,
        x(k+1) - Ax(k) - Bu(k)  \\
        &- W\phi(x(k)) \in \calV,\, \forall k \in \{0, \dots, t-1\}
    \big\}.
\end{aligned}\end{equation}
When $\calV$ is a hyperbox, this estimator maintains independent feasible sets for each row of $W$ and can be updated recursively in time with polytopical set intersections by rewriting \cref{eq:est-setmember-fps} in terms of the row-wise vectorization of $W$. 
Clearly, $\Theta(t) \subseteq \Theta(t-1)$. As is common practice in the literature \cite{bujarbaruah_semi-definite_2020}, we propose generating a point estimate of the parameters as the \emph{Chebyshev center} of the feasible parameter set:
\begin{align}\label{eq:est-chebycenter}
    \what(t) = \arg \min_{\what}\max_{W \in \Theta(t)} \|\what - W\|_F.
\end{align}
By definition, this approach minimizes the worst-case error of the point estimates and is typically straightforward to compute \cite{boyd_cvx_2004}. Denoting the Chebyshev radius for the feasible parameter set associated with the $i$-th row of $W$ as $r_i(t) = \min_{w} \max_{w_i \in \Theta_i(t)} \|w - w_i\|_2$, we take the confidence interval on $\hat w_i(t) - w_i$ as $\calW_i(t) = \{\tilde w_i : \|\tilde w_i\|_2 \leq r_i(t)\}$. 

By definition, since $\Theta(t) \subseteq \Theta(t-1)$, the Chebyshev radii must be decreasing over time: $r_i(t) \leq r_i(t-1)$. Therefore, a set-membership estimator with point estimates as the Chebyshev center satisfies \cref{as:shrinking}. 
% In fact, since a set membership approach ensures the true parameter lies in a shrinking, known set, this estimator has somewhat stronger properties than we require. 

%------------------------------------------------
\textbf{Recursive Bayesian Linear Regression (BLR).}
In the case of Bayesian estimation, we can generate confidence intervals directly from the posterior distribution over parameters if we know the disturbance distribution. We outline this approach under a simple, standard assumption.
\begin{assumption}\label{as:est-subgaussian-noise}
We assume that each entry of the process noise is bounded $v(t) = [v_1(t), \dots,v_n(t)]^\top \in \calV = \{v : |v| \leq \sigma_i\}$, and that each entry $v_i(t)$ is independent of the others. Hence, $v_i(t)$ is sub-Gaussian with variance proxy $\sigma_i^2$.  
\end{assumption}
Under \cref{as:est-subgaussian-noise}, we can essentially treat the noise as both normally distributed for convenient analysis and provide safety guarantees for the algorithm proposed in \cref{sec:approach}. If we place subjective priors over the rows of $W$ of the form $w_i \sim \mathcal{N}(\hat{w}_i(0), \sigma_i^2 \Lambda_i^{-1}(0))$, then the resulting posterior parameter distribution at time $t$ is also Gaussian, $w_i \sim \mathcal{N}(\hat{w}_i(t), \sigma_i^2 \Lambda_i^{-1}(t))$. We then use the mean of the posterior---also corresponding to the maximum a posteriori (MAP) estimate---as a point estimate for control: $\hat{f}(x,t) = \what(t)\phi(x)$. Defining the measurement and prediction at time $t$ as $y(t) := x(t+1) - Ax(t) - Bu(t)$ and $\hat{y}(t):= \what(t) \phi(t)$, the MAP estimate for each row can then be updated with constant complexity in time using the recursive updates   
\begin{align}
    \hat{w}_i(t+1) &= \hat{w}_i(t) - \frac{(\hat{y}_i(t)- y_i(t))\phi(t)^\top \Lambda_i^{-1}(t)}{1 + \phi(t)^\top \Lambda_i^{-1}(t) \phi(t)}  \label{eq:predict}\\[1em]
    \Lambda_i^{-1}(t+1) &= \Lambda_i^{-1}(t) - \frac{\Lambda_i^{-1}(t)\phi(t) \phi(t)^\top \Lambda_i^{-1}(t)}{1 + \phi(t)^\top \Lambda_i^{-1}(t) \phi(t)}, \label{eq:update}
\end{align}
where we write each entry of $\hat{y}(t)$ as $\hat{y}_i(t) = \hat{w}_i(t)^\top \phi(t)$ and $\phi(t) := \phi(x(t))$. 
We can recover the frequentist ordinary least-squares estimator if we assume a flat prior \cite{deisenroth2020mathematics}, which requires the availability of some amount of prior data to yield the initial values $\what(0)$, $\Lambda(0)$. 
        
Taking a risk tolerance of $\delta \in (0,1)$, we could naively define the confidence interval for the $i$-th row of $\what(t)$ as 
\begin{align}\label{eq:est-lstsq-naive}
    \calW^{\mathrm{naive}}_i(t) := \{\tilde w_i \in \R^n : \tilde w_i^\top \Lambda_i(t) \tilde w_i \leq \sigma_i^2 \chi^2_n(1-\frac{\delta}{n})\},
\end{align}
where $\chi_n^2(1-\delta)$ is the $1 - \delta$ quantile of the chi-square distribution with $n$ degrees of freedom. However, the confidence interval in \cref{eq:est-lstsq-naive} does not capture the fact that we want to certify the safety of the policy for all time with high probability. We cannot achieve this with a single confidence interval of a point estimate at time $t$, as \cref{eq:est-lstsq-naive} ignores the correlations between the model estimates over time. For robust control, we instead desire confidence intervals such that
\begin{equation}\label{eq:est-lstsq-conf}
    \what(t) - W \in \calW(t),\ \forall t \in \mathbb{N}_{\geq 0},
\end{equation}
with probability at least $1 - \delta$.

Recent work applied a Martingale argument originating from the Bandits literature to generate such confidence intervals \cite{lew_safe_2020} by scaling the naive confidence intervals by a time-varying parameter\footnote{Subject to assumptions on the calibration of the prior, for which we refer the reader to \cite[Assumption 3]{lew_safe_2020}. This assumption is trivially satisfied for flat priors, and we assume this assumption is satisfied for subjective priors.}. 

\begin{theorem}\label{thm:thomas}
\cite[Thm~1]{lew_safe_2020} For the recursive Bayesian linear filter \cref{eq:predict},\cref{eq:update}, we have the estimation error $\hat{w}_i(t) - w_i \in \calW_i(t)$ for all $t\geq 0$ with probability at least $1-\delta$, where
\begin{equation}\label{eq:bls-conf}
    \calW_i(t) 
    := \{\tilde w_i : (\tilde w_i^\top \Lambda_i(t) \tilde w_i)^{\frac{1}{2}} \leq \sigma_i \beta_t(\delta /n)\},
\end{equation}
with $\beta_t(\delta)$ equal to
\begin{equation}
    \sqrt{2 \log \Big(\frac{\det(\Lambda_i(t))^{1/2}}{\delta \det(\Lambda_i(0))^{1/2}}\Big)} + \sqrt{\frac{\lambda_\mathrm{max}(\Lambda_i(0))}{\lambda_\mathrm{min}(\Lambda_i(t))} \chi^2_n(1-\delta)}.
\end{equation}
\end{theorem}

The confidence intervals resulting from \cref{thm:thomas} unfortunately do not immediately satisfy \cref{as:shrinking} without a \emph{persistence of excitation} or \emph{active exploration} assumption as is made in \cite{mania2020active}. A simple workaround is to update the estimate \cref{eq:predict} fed into to the controller only when the associated confidence intervals \cref{eq:bls-conf} have shrunk, effectively disregarding new data until the system has been excited sufficiently. This approach was shown to perform well in practice \cite{koller_learning-based_2018}. Still, future work should explore strategies to guarantee confidence intervals constructed using \cref{thm:thomas} (or other equivalent results) satisfy \cref{as:shrinking} more naturally.  

\section{Experiments}\label{sec:exp}
To benchmark our adaptive, robust MPC (ARMPC) algorithm that \emph{matches} as much of the uncertainty as possible, we compare it against the normative approach in ARMPC; we estimate the range of values that the uncertainty $f$ can take and naively treat it as a disturbance using tube MPC as in \cite{bujarbaruah_adaptive_2018, bujarbaruah_semi-definite_2020, soloperto_learning-based_2018, KohlerAndinaEtAl2019}. To the best of our knowledge, such algorithms have not been proposed for uncertain terms that satisfy \cref{as:bounded-feat} in the literature. Therefore, we apply the tube MPC strategy in \cite{goulart_optimization_2006} online using the disturbance set $\calV'(t) = \calFhat(t) \oplus \calV$, since \cref{cor:unmatched-support} implies that both $f(x(t)) + v(t) \in \calV'(t)$ and $\calV'(t) \subseteq \calV'(t-1)$ for all $t\geq0$. 

\begin{figure}[t]
    \centering
        \includegraphics[width=\linewidth]{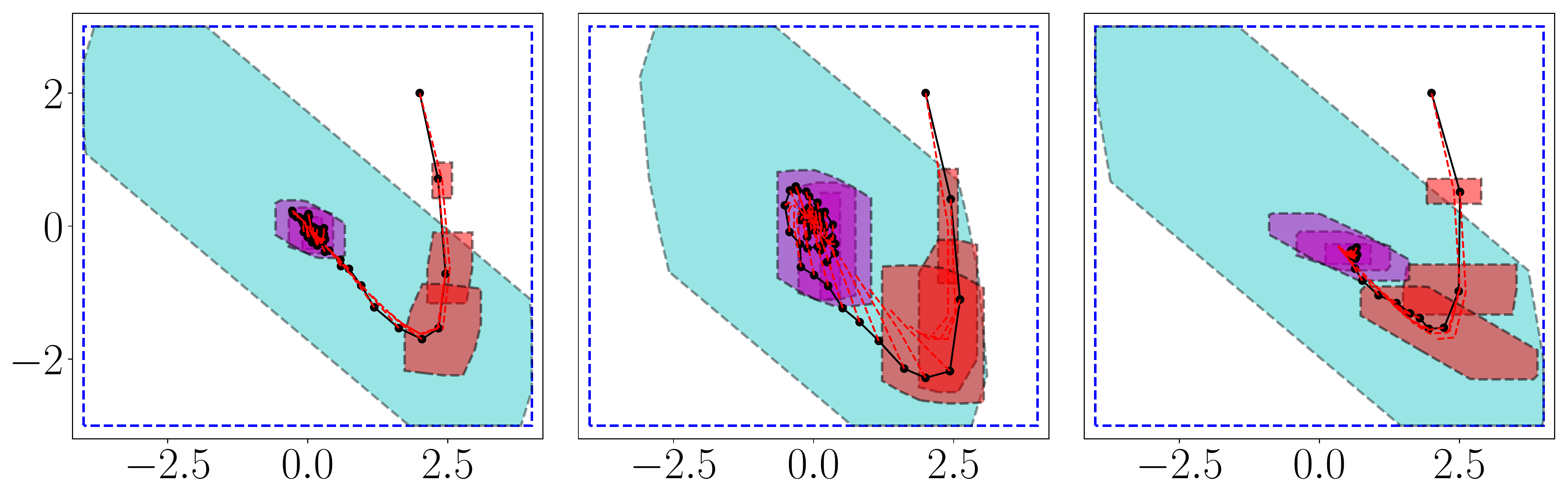}
    \caption{Closed-loop trajectories of different systems. The true system evolution is in black, the predicted nominal trajectories are red, as are the predicted reachable sets at $t=0$. The predicted reachable sets at $t=50$ are in purple. The dark blue line indicates the state constraints. The terminal invariant $\mathcal{O}(50)$ is shown in cyan. Left: Adaptive CE MPC on the system with \emph{matched} uncertainty \cref{eq:exp-matched-dyn}. Middle: Benchmark Adaptive MPC on the \emph{matched} system \cref{eq:exp-matched-dyn}. Right:  Adaptive CE MPC on the system with \emph{unmatched} uncertainty \cref{eq:exp-unmatched}.}  
    \label{fig:trajs}
\end{figure}

\textbf{Simple Matched System.}\label{sec:exp-matched}
We illustrate the properties of our algorithm on a double-integrator system
\begin{equation}\label{eq:exp-matched-dyn}
    x(t+1) = \begin{bmatrix}1 & 0.2 \\ 0 & 1 \end{bmatrix}x(t) + \begin{bmatrix}0 \\ 1\end{bmatrix}u(t) + f(x(t)) + v(t),
\end{equation}
a typical example in the MPC literature that represents simplified second-order dynamics \cite{borrelli_predictive_2017}. First, we consider the \emph{matched} uncertainty $f(x) =  [0, w_1]^\top \tanh([0,1]x)$. We estimate the true parameter, $w_1 = 0.5$, online to improve performance. We take the disturbance as an isotropic Gaussian with $\sigma^2 = 5 \times 10^{-3}$ clipped at its 95\% confidence intervals.
In addition, the system is subject to the state and input constraints ${(-4,-3) \preceq x \preceq (4,3)}$ and $-2 \leq u \leq 2$, respectively. We regulate the system to the origin from ${x(0) = [2, 2]^\top}$ while minimizing the quadratic cost function $h(x, u) = x^\top Q x + u^\top Ru$ over a horizon of length $N=3$. We take $Q = I_2$, $R = 1$. We use the BLR estimator with a flat prior and collect $k = 45$ data points of the system evolution in a unit box near the origin to form an initial estimate of the model parameters and construct $\mathcal{O}(t)$ using Alg. 10.5 in \cite{borrelli_predictive_2017}. 

We plot the closed-loop system evolution in \cref{fig:trajs} (left, middle). Our adaptive certainty-equivalent MPC algorithm is able to effectively control the system. In contrast, the benchmark ARMPC can only react to the learned dynamics after $f$ enters the system as a disturbance, resulting in considerably larger closed-loop oscillations and uncertainty on the predicted trajectory. In \cref{fig:trajs}, we also plot the reachable sets associated with the predicted trajectory at the first and last timesteps of the control task. The reachable sets show that our adaptive MPC resolves prediction uncertainty in the system since they shrink over time. Moreover, online learning has little consequence for the benchmark, as increasing the confidence in the model does not significantly reduce the estimated range of values that the nonlinear function takes.

In addition, we compare the asymptotic performance of our algorithm with the benchmark ARMPC as a function of~$w_1$, controlling the magnitude of the nonlinearity. To do this, we set the number of data points used to generate the initial model estimate to~$k=10^3$ and collect trajectory rollouts for various values of~$w_1$ from the fixed initial condition. As \cref{fig:cost-envelopes} (left) shows, our ARMPC is feasible for disturbances more than twice the magnitude of those the benchmark ARMPC can tolerate for the given initial condition. In addition, the realized control cost does not differ significantly from the benchmark for values of $w_1$ when both controllers are feasible.  

Finally, we illustrate how our algorithm tolerates larger dynamic uncertainty by comparing the size of the feasible envelope (i.e., the set of initial conditions for which the MPC problem is feasible) as a function of~$w_1$. We set the number of data points to inform our prior to a modest~$k=50$, grid the state-space, and take the feasible region as the convex hull of the initial conditions for which the MPC problem is feasible. Then, we estimate the percentage of states ~$x_0 \in \calX$ in the feasible envelope as the ratio of volumes between the feasible envelope and the state space~$\calX$, illustrated with solid lines in \cref{fig:cost-envelopes} (right). Our adaptive MPC algorithm can tolerate much larger disturbances than the benchmark ARMPC. In these experiments, the feasible envelope of the benchmark becomes empty when the maximal robust invariant is null ($\mathcal{O}(t)=\emptyset$), indicating that there is no subset of~$\calX$ in which the LQR policy associated with~$h$ results in provably safe behavior~\cite{borrelli_predictive_2017}. Hence, \cref{fig:cost-envelopes} highlights the fragility of existing ARMPC approaches under large disturbances.

\begin{figure}[t]
    \centering
        \includegraphics[width=\linewidth]{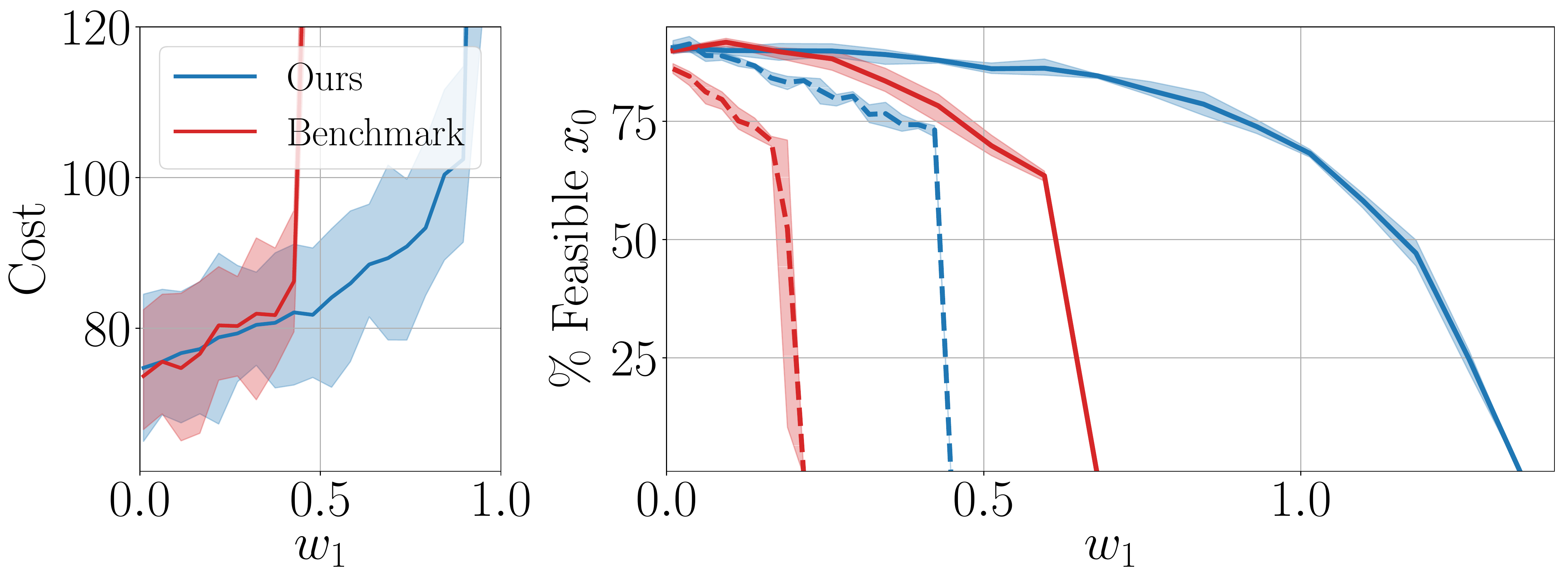}
    \caption{Left: Closed-loop realized trajectory cost for the matched system \cref{eq:exp-matched-dyn} as a function of~$w_1$. Exploding cost indicates infeasibility. The error bars indicate~${2\sigma}$ bounds. Right: Solid lines indicate size of feasible envelopes as a function of $w_1$ for a matched system. Dashed lines indicate the size of the feasible envelopes for the unmatched system \cref{eq:exp-unmatched} as a function of the magnitude of the unmatched dynamics $w_1$ with $w_2=0.5$.}
    \label{fig:cost-envelopes}
\end{figure}

\textbf{Simple Unmatched System.} We now extend the simulations of the simple matched system to the unmatched case to understand the effect of additive nonlinear terms that cannot be canceled from the dynamics. We keep the nominal dynamics identical to \cref{eq:exp-matched-dyn} and take the nonlinear dynamics
\begin{equation}\label{eq:exp-unmatched}
    f(x) = \frac{1}{\sqrt{2}}[w_1 \sin(4x_1), w_2 \tanh(x_2)]^\top,
\end{equation}
where~$w_1$ and~$w_2$ are unknown parameters. Similar to the previous experiments, we initialize the model with~$k=45$ data points sampled around the origin. In this example, the certainty equivalent policy \cref{eq:unmatched-ce-law} can only compensate for the second component of the nonlinear dynamics \cref{eq:exp-unmatched}. We set $w_1 = 0.2$ and $w_2=0.3$. In \cref{fig:trajs} (right), the size of the reachable sets increases if we simulate the system with the unmatched dynamics \cref{eq:exp-unmatched}. The benchmark ARMPC was infeasible from this initial condition, showing that our method still outperforms the benchmark. Next, we fix $w_2=0.5$ and vary $w_1$ to understand the impact of an estimated, unmatched dynamics component. \cref{fig:cost-envelopes} (dashed, right) shows that in our experiment, matching as much of the nonlinear dynamics as possible allows us to handle unmatched dynamics of about twice the magnitude as the benchmark. We conclude that our method is a more effective strategy even if the uncertainty is unmatched. Naturally, \cref{fig:cost-envelopes} (right) also shows that the benefit of our method diminishes as the proportion of the nonlinear dynamics~$f(x(t))$ in $\mathrm{Range}(B)$ becomes smaller.

\textbf{Planar Quadrotor.}\label{sec:exp-quadrotor}
Finally, we simulate a simplified example of a quadrotor in a windy environment. The force field induced by the wind varies spatially, modeling real-world scenarios such as down-wash from another quadrotor. We consider a planar version of the quadrotor for simplicity \cite{tedrake_underact}. We model the wind disturbance as incident at a fixed angle with a velocity that drops off according to an inverse square exponential normal to the direction of incidence. We linearize the dynamics around~$x=0$,~$u=\frac{mg}{2}[1,1]^\top$ and discretize the simulation using Euler's method. We only set constraints on the pose of the drone~$(x,y,\theta)$. Its linear and angular velocities are unconstrained. The quadrotor is an underactuated system, and therefore the discretized simulation has unmatched dynamics terms. Still, a drone controller can always match disturbance forces along the~$y$ axis. We take a Bayesian approach and model the unknown wind disturbance as a function of~$d=20$ normalized random Fourier features~$\phi_i(x) = \frac{1}{\sqrt{d}}\cos(\alpha_i^\top x + \beta_i)$ where~$\alpha\sim\mathcal{N}(0,I_n)$ and~$\beta_i\sim \mathrm{U}[0,2\pi]$. To calibrate the Bayesian prior, we first select a zero-mean prior with a variance chosen such that the confidence interval reflects a conservative bound on the wind disturbance. We then calibrate the Bayesian prior using~$k=100$ historical wind samples. 

The support of the estimated wind disturbance was too large for our benchmark to be feasible in our experiments, so we compare our approach with a naive tube MPC that does not account for the wind disturbance at all. As shown in \cref{fig:drone} (left), if the wind disturbance is axis aligned, our ARMPC learns to match the wind forces and reaches the origin quickly. In contrast, the naive tube MPC misses the origin and drifts significantly. In addition, if we set the angle of incidence of the wind as~$\theta_w = 22.5^\circ$, \cref{fig:drone} (right) shows that our approach still achieves decent control performance. This is because our certainty equivalent controller \cref{eq:unmatched-ce-law} cancels the~$y$-component of the disturbance and converges to a small steady-state offset in the~$x$-direction. In contrast, the benchmark ARMPC could not guarantee safety for any of the drone tasks, and a naive unsafe tube MPC that does not consider the wind disturbance at all performs poorly. Moreover, model adaptation reduced the realized closed-loop cost in the simulation \cref{fig:drone} (left) by about $12\%$ over 10 episodes from $J^1 = 2570$ to $J^{10} = 2259$. 

\begin{figure}[t]
    \centering
        \includegraphics[width=\linewidth]{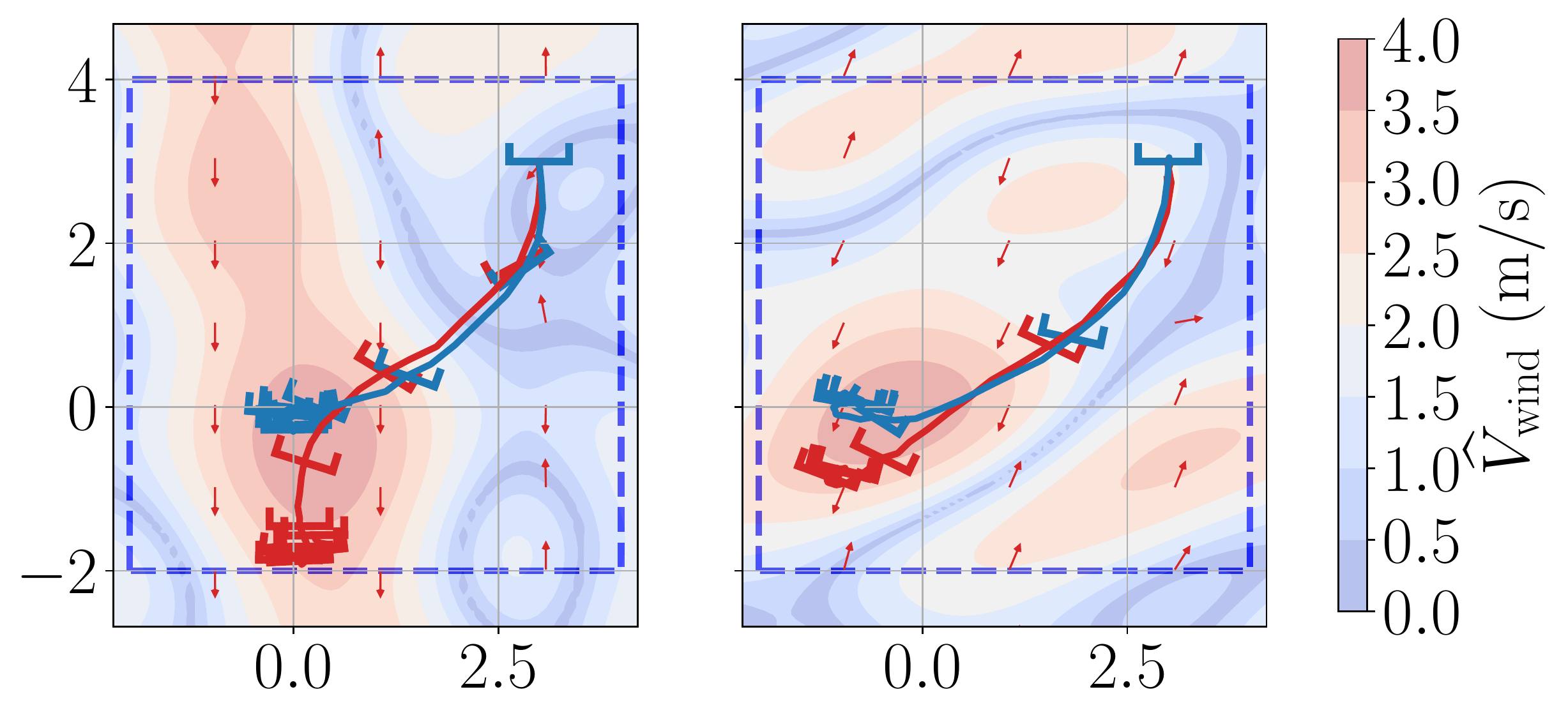}
    \caption{Learned trajectories of our adaptive MPC (blue) compared to a naive tube MPC (red) in the $xy$-plane for a simulated planar quadrotor. The icons show the orientation of the quadrotor over time. The contours and colorbar indicate the \emph{learned} wind speed in~$\mathrm{m/s}$. Left: Wind comes straight from above. Right: Wind comes at $\theta_w = 22.5^\circ$.}
    \label{fig:drone}
\end{figure}

\section{Conclusions and Future Work}\label{sec:conc}
The simulations in \cref{sec:exp} show that our method can achieve substantial performance improvements compared to existing approaches, even when a certainty equivalent controller cannot cancel significant, unmatched components of the nonlinear dynamics. We conclude that we can reduce the conservatism of robust MPC approaches by extending certainty equivalent control laws from classical adaptive control. As a result, our method tolerates nonlinear terms of larger magnitude in the dynamics.

Since our control algorithm allows for adaptation laws based on statistical techniques that are more robust to misspecification or outliers than set membership estimation, future work should extend our simulations to hardware experiments. In addition, our method requires that the features are known a priori, an assumption that future work could relax by, for example, applying meta-learning algorithms \cite{harrison_meta-learning_2020}. 
%*****************************************************************************
\section*{Acknowledgements}
We thank Monimoy Bujarbaruah for his thoughtful comments on an early manuscript. This research was supported in part by the NSF via CPS award \#1931815, NASA ULI grant \#80NSSC20M0163, and an Early Stage Innovations grant. S.M.R. and J.H. were also supported in part by NSERC. This article solely reflects our own opinions and conclusions, and not those of any NSF, NASA, or NSERC entity.

%%%%%%%%%%%%%%%%%%%%%%%%%%%%%%%%%%%%%%%%%%%%%%%%%%%%%%%%%%%%%%%%%%%%%%%%%%%%%%%%

% \addtolength{\textheight}{-3cm}   % This command serves to balance the column lengths

%%%%%%%%%%%%%%%%%%%%%%%%%%%%%%%%%%%%%%%%%%%%%%%%%%%%%%%%%%%%%%%%%%%%%%%%%%%%%%%%
%%%%%%%%%%%%%%%%%%%%%%%%%%%%%%%%%%%%%%%%%%%%%%%%%%%%%%%%%%%%%%%%%%%%%%%%%%%%%%%%

\bibliographystyle{unsrt} 
\bibliography{strings,references}  
\end{document}